\documentclass[aps,pra,preprint,superscriptaddress,amsmath,amssymb,showpacs]{revtex4-1}
\usepackage{bm}
\usepackage{amsmath,amssymb}
\newtheorem{theorem}{Theorem}
\newtheorem{proposition}[theorem]{Proposition}
\newtheorem{lemma}[theorem]{Lemma}
\newtheorem{corollary}[theorem]{Corollary}
\newenvironment{proof}{\noindent{\it Proof. \hskip0pt}}
                     {$\square$\par\medskip}
\begin{document}


\title{Optimal indecomposable witnesses without extremality\\ as well as spanning property}

\author{Kil-Chan Ha}
\author{Hoseog Yu}
\affiliation{Faculty of Mathematics and Statistics, Sejong University, Seoul 143-747, Korea}
\date{\today}

\begin{abstract}
One of the interesting problems on optimal indecomposable entanglement witnesses is whether there exists an optimal indecomposable witness which neither has the spanning property nor is associated with extremal positive linear map. Here, we answer this question negatively by examining the extremality of the positive linear maps constructed by Qi and Hou [J. Phys. A {\bf 44}, 215305 (2100)]. 
\end{abstract}

\pacs{03.67.Mn,03.65.Ud}
\keywords{positive linear map, optimal entanglement witness, extremal positive linear map, spanning property, positive semidefinite form}

\maketitle


\section{Introduction}
A most general approach for distinguishing entanglement from separable states may be a criterion based on the notion of entanglement witness \cite{horo-1,terhal}. 
A Hermitian operator $W$ acting on a complex Hilbert space $\mathcal H\otimes \mathcal K$ is called an   entanglement witness (EW) if $W$ is not positive and ${\rm Tr}(W\rho)\ge 0$ holds for all separable states $\rho$. Thus, if $W$ is an EW, then there exists an entangled state $\rho$ such that ${\rm Tr}(W\rho)<0$ (In this case, we say that $\rho$ is dectected by $W$).  It is well known \cite{horo-1} that a state is entangled if and only if it is dectected by some entanglement witness.

For  finite dimensional Hilbert spaces, this criterion is closely connected to the duality theory \cite{eom-kye} between positivity of linear maps and separability of block matrices, through the Jamio\l kowski-Choi
isomorphism \cite{choi75-10, jami}. That is, a self-adjoint block matrix $W\in M_m\otimes M_n$ is an EW if and only if there exists a positive linear map that is not comletely positive  $\Phi: M_m\to M_n$ such that
\[
W=\frac 1m C_{\Phi}=\frac 1m\sum_{i,j=1}^m |i\rangle\langle j|\otimes \Phi(|i\rangle \langle j|),
\]
where $M_n$ denotes the $C^*$-algebra of all $n\times n$ matrices over the complex field $\mathbb C$ and the block matrix $C_{\Phi}$ is the Choi matrix of $\Phi$. We denote $W_{\Phi}=1/m\, C_{\Phi}$ for the entanglement witness associated with the positive map $\Phi$.

It is well known that decomposable positive linear maps give decomposable entanglement witnesses which take general form $W=P+Q^{\Gamma}$, where $P,Q\ge 0$ and $Q^{\Gamma}$ denotes the partial transpose of $Q$. If a given witness can not be written in this form, we call it indecomposable. Of course, indecomposable EWs are associated to indecomposable positive linear maps \cite{lew00, lew01,hakyepla}.

To characterize the set of EWs, the notion of optimality is important. An entanglement witness which detects a maximal set of entanglement is said to be optimal, as was introduced in \cite{lew00}. Since every witness can be optimized \cite{lew00}, optimal EWs are sufficient to detect all the entangled states.
So, it is significant to characterize the set of optimal EWs. Although there was a considerable effort in this direction  \cite{lew01,acin01,sarbicki,sperling,sz,ssz,cw,atl,asl,kye-dec-oew,ha-kye-pra11,xia,ha-kye-os11,kye_ritsu,hakye12,hakye12pra,hakye12prl,ha12}, complete characterization and classification of optimal EWs are far from satisfactory. 

In Ref.~\cite{lew00}, it was shown that: (1) $W$ is an optimal EW if and only if $W-Q$ is no longer an EW for an arbitrary positive semi-definite matrix $Q$; (2) $W$ is an optimal EW if $W$ has spanning property, that is $\mathcal P_{W}=\{|\xi, \eta\rangle \in \mathbb C^m\otimes \mathbb C^n\,:\, \langle \xi,\eta |W|\xi,\eta\rangle=0\}$ spans the whole space $\mathbb C^m\otimes \mathbb C^n$. From the criterion (1), we see that EW associtated to an extremal positive linear map is optimal. By an  extremal positive linear map, we mean a positive linear map which generates an extremal ray of the convex cone consisting of all positive linear maps.  That is, a positive linear map $\phi$ is said to be extremal if $\phi =\phi_1+\phi_2$ with positive linear maps $\phi_i$, should imply $\phi_i=\lambda_i \phi$ with nonnegative real numbers $\lambda_i$.
In the case of indecomposable EW, the Choi map \cite{choi75, choi-lam} and its variations \cite{kye_ext,osaka,ckl,ha-kye-os11,sengupta} are extremal and give rise to  optimal EWs. Although the extremality of a positive linear map gives us a sufficient condition for the optimality of the associated EW, it is very difficult to check whether a positive linear map is extremal. On the other hand, the criterion (2) is very pratical for checking optimality of witnesses. In fact, almost all known optimal EWs are investigated by this criterion. (See the Refs.~\cite{cw,ha-kye-pra11} and references therein). However, the spanning property is also not a necessary condition for optimality of EW. In fact,  the extremal Choi map \cite{choi75, choi-lam} introduces an optimal EW that have no spanning property. See the Ref.~\cite{asl} for examples of optimal decomposable EWs without spanning property.

Recently, in order to examine optimality of EW without spanning property, two kinds of methods are provided with examples of optimal indecomposable EWs which have no spanning property. Xia and Hou's approach is based on reinterpretation of optimal EW in terms of positive map \cite{xia,xia11}. The first author and Kye \cite{hakye12pra} checked optimality by examining the facial structure of the convex body containing the positive linear map associated with the target EW.  It remains to be shown whether these examples in \cite{xia,hakye12pra} are associated with extremal positive linear maps.
To the best of the author's knowledge, only known examples of optimal indecomposable EW without spanning property are associated with  positive linear maps which are variations of the Choi map. Then these postive linear maps are turned out to be extremal besides examples in \cite{xia,hakye12pra}. Therefore, it is natural to ask whether every optimal indecomposable EW without spanning property is associated with extremal positive linear map. The primary aim of this paper is to clarify this point.

For this purpose, we study the extremality of the indecomposable positive linear map $\Phi^{(n,k)}$ constructed by Qi and Hou \cite{xia11}. Then, we answer this question negatively by showing that $\Phi^{(n,k)}$ is not extremal whenever $n$ and $k$ have common divisors greater than 1, that is, $\gcd(n,k)>1$. Note that the optimality of associated entanglement witness $W_{\Phi^{(n,k)}}$ with no spanning property is already known \cite{xia}. It was also observed \cite{ha12} that $W_{\Phi^{(n,k)}}$ is a PPTES entanglement witness \cite{ha12} (that is, nd-OEW in the sense of \cite{lew00}) since $W_{\Phi^{(n,k)}}^{\Gamma}$ has the spanning property. See the Ref. \cite{hakye12} for PPTES entanglement witness. Consequently, $W_{\Phi^{(n,k)}}$ (with $\gcd(n,k)>1$) becomes the first example of optimal indecomposable EW, which neither has  spanning property nor is associated with extremal positive linear map. For the case of $\gcd(n,k)=1$, we try to show that $\Phi^{(n,k)}$ is extremal. First, we show that $\Phi^{(n,k)}$ is extremal if and only if $\Phi^{(n,1)}$ is extremal when $\gcd(n,k)=1$. Then we show that $\Phi^{(4,1)}$ and so $\Phi^{(4,3)}$ are indeed extremal. For general $n$, we think that the extremality of $\Phi^{(n,1)}$ can dealt with similarly. Our approach to tackle extremality is based on Choi and Lam's method \cite{choi75,choi-lam,osaka,choi80} using the correspondence between positive semidefinite biquadratic forms and positive linear maps. Through the decomposition of biquadratic form corresponding $\Phi^{(n,n/2)}$, we also reprove that $\Phi^{(n,n/2)}$ is decomposable when $n$ is even integer greater than $2$.

In the next section, we recall the positive linear maps $\Phi^{(n,k)}$ and explain how to check the extremality of those maps according to Choi and Lam's method \cite{choi75,choi-lam,osaka,choi80}. After we explore some extremal positive semidefinite forms in Section 3, we analyze the extremality of $\Phi^{(n,k)}$ in the last section. 

Throughout this note, $\sigma_{k}:\{1,2,\cdots,n\}\to \{1,2,\cdots,n\}$ denotes the permutation defined by 
$\sigma_k(i)=i+k\ {\rm mod}\ n$.

\section{Preliminaries}\label{map_form}
First, we recall \cite{xia11} the positive linear map $\Phi^{(n,k)}:M_n\to M_n$ for each $k=1,2,\ldots,n-1$ defined by 
\begin{equation}\label{map}
\Phi^{(n,k)}([a_{ij}])={\rm diag}(b_1,b_2,\cdots,b_n)-[a_{ij}]
\end{equation}
for $[a_{ij}]\in M_n$, 
where $b_{i}=(n-1)a_{ii}+a_{\sigma_k(i),\sigma_k(i)}$ for each $i=1,2,\ldots,n$ $(n\ge 3)$.

X. Qi and J. Hou \cite{xia11}  showed that $\Phi^{(n,k)}$ are indecomposable positive linear maps whenever 
either $n$ is odd or $k\neq n/2$. They also showed \cite{xia} that the associated EWs $W_{\Phi^{(n,k)}}$ are optimal EWs which have no spanning property whenver $k\neq n/2$, and $W_{\Phi^{(n,n/2)}}$ is decomposable and not optimal when $n$ is an even integer greater than $2$.  Recently, it was shown \cite{ha12} that $W_{\Phi^{(n,k)}}$'s are indeed optimal PPTES witnesses whenever $k\neq n/2$ (that is, nd-OEW in the sense \cite{lew00}). Therefore, $W_{\Phi^{(n,k)}}$ detects a maximal set of entangled states with positive partial transposes in the sense \cite{hakye12}.
Especially,  $\Phi^{(3,1)}$ and $\Phi^{(3,2)}$ are  extremal Choi maps \cite{choi75}. So these maps can be considered as extensions of  extremal Choi map in the $n$-dimensional cases. Thus, we may expect that these maps are extremal. But, in general, these maps are not extremal.
Although we can show that $\Phi^{(4,1)}$ and $\Phi^{(4,3)}$ are extremal, $\Phi^{(4,2)}$ is not extremal since $W_{\Phi^{(4,2)}}$ is not optimal.  We will also show that $\Phi^{(n,k)}$ is not extremal if $\gcd(n,k)\neq 1$. Note that $W_{\Phi^{(n,k)}}$ is still optimal in the case of $\gcd(n,k)\neq1$ as long as $k\neq n/2$. This is the point of this work.

We note that $\Phi^{(n,k)}$ maps $M_n(\mathbb R)$ into inself. Therefore, we can use Choi and Lam's method \cite{choi-lam} (see also Ref.~\cite{osaka}) to check the extremality of $\Phi^{(n,k)}$. 
For each $n\ge 4$ and $k=1,2,\cdots,n-1$, we define positive semidefinite biquadratic forms $B_{\Phi^{(n,k)}}$ by 
\begin{equation}\label{biquad_form}
\begin{aligned}
B_{\Phi^{(n,k)}}\left(\begin{array}{cccc} x_1 & x_2 & \cdots & x_n \\y_1 & y_2 & \cdots & y_n \end{array}\right)
:&=y^{\rm t}\left[ \Phi^{(n,k)}(x x^{\rm t})\right] y\\
&=(n-2)\sum_{i=1}^n x_i^2 y_i^2 +\sum_{i=1}^n x_{\sigma_k(i)}^2 y_i^2 -2 \sum_{1\le i<j\le n}x_iy_ix_jy_j,
\end{aligned}
\end{equation}
where $x=(x_1,x_2,\cdots,x_n)^{\rm t},\,y=(y_1,y_2,\cdots,y_n)^{\rm t}\in \mathbb R^n$.
(By the same way, we can define the positive semidefinite biquadratic form $B_{\phi}$ corresponding to a positive linear map $\phi$.)
Let $\mathcal P_{n,m}$ be the set of all positive semidefinite (psd) real forms in $n$ variables of degree $m$. Then each $B_{\Phi^{(n,k)}}$ belongs to $\mathcal P_{2n,4}$ since $\Phi^{(n,k)}$ is a positive linear map.
A form $F\in \mathcal P_{n,m}$ is said to be extremal if $F=F_1+F_2$, $F_i\in P_{n,m}$, should imply $F_i=\lambda_i F$ with nonnegative real numbers $\lambda_i$.
If we write $\mathcal E(\mathcal P_{n,m})$ for the set of all extremal positive semidefinite forms in $\mathcal P_{n,m}$, an elementary result in the theory of convex bodies shows that $\mathcal E(\mathcal P_{n,m})$ spans $\mathcal P_{n,m}$. 
It is well known \cite{choi-lam,osaka} that if a positive linear map $\phi:M_n\to M_n$ maps $M_n(\mathbb R)$ into itself, then the corresponding biquadratic form $B_{\phi}\in \mathcal E(\mathcal P_{2n,4})$ implies that $\phi$ is extremal in the convex cone consisting of all positive linear maps. Therefore, to show the extremality of  $\Phi^{(n,k)}$, it suffices to prove $B_{\Phi^{(n,k)}}\in \mathcal E(\mathcal P_{2n,4})$.

We also note that a psd biquadratic form $B$ gives rise to a positive linear map $\phi$ such that $B=B_{\phi}$.
Let $B(X:Y)$ be a biquadratic form where $X=(x_1,x_2,\cdots,x_m)^{\rm t}\in \mathbb R^m$ and 
$Y=(y_1,y_2,\cdots,y_n)^{\rm t}\in \mathbb R^n$.
Since this biquadratic form can be considered as a quadratic form with respect to each variable $Y$ (as well as X), we can write it in the form $\langle Y|S_X|Y\rangle$ where $S_X\in M_n$ is a symmetric matrix. Thus we get a map sending each one-dimensional projection $XX^{\rm t}\in M_m$ to $S_X$. Using linearity and hermiticity, we can extend it to a map which preserve hermiticity. It was shown by Choi that, given any positive semidefinite form, this corresponding linear map is a positive linear map 
\cite{choi75,choi80}. 

For example, for a given psd biquadratic form $B(X:Y)=(x_1y_3-x_3y_1)^2$ with $X,\,Y\in \mathbb R^3$, we get a symmetric matrix $S_X$ of the form
\[
S_X=\left( \begin{array}{ccc} x_3^2 & 0 & -x_1x_3 \\ 0 & 0 & 0 \\-x_3 x_1 & 0 & x_1^2\end{array}\right).
\]
Consequently, we obtain a positive (in fact, completely copositive) linear map $\phi:M_3\to M_3$ defined by 
\[
\phi([a_{ij}])=a_{33}|1\rangle \langle 1|-a_{13}|1\rangle \langle 3|-a_{31}|3\rangle \langle 1|+a_{11}|3\rangle \langle 3|=V[a_{ij}]^{\rm t} V^{\rm t}
\]
for $[a_{ij}]\in M_3$ and $V=|3\rangle \langle 1|-|1\rangle \langle 3|$. 
We will use this correspondence between psd biquadratic forms and positive linear maps to show that if a biquadratic form $B_{\Phi^{(n,k)}}$ is decomposed into the sum of psd biquadratic forms then the corresponding map $\Phi^{(n,k)}$ is not extremal. 

\section{Extemal positive semidefinite forms}

In this section, we explore some psd forms needed to show the extremality of $B_{\Phi^{(4,1)}}$.
For the positive linear map $\Phi^{(4,1)}$, we define a quaternary octic psd form $O_{(4,1)}(x,y,z,w)$ by 
\[
\begin{aligned}
O_{(4,1)}(x,y,z,w):=&B_{\Phi^{(4,1)}}\left( \begin{array}{cccc}
                                              yzw & zwx & wxy & xyz\\
                                                x   &   y   &    z   &  w
                                               \end{array}
                                      \right)
\\
=& x^4 z^2 w^2 +y^4 x^2 w^2 +z^4 x^2 y^2 +w^4 y^2 z^2-4x^2 y^2 z^2 w^2.
\end{aligned}
\]
We show that $O_{(4,1)}(x,y,z,w)$ is extremal in $\mathcal P_{4,8}$.
Assume that there is a psd form $F\in \mathcal{P}_{4,8}$ such that $0 \le F \le O_{(4,1)}$. 

We note that  the only possible monomial of $F$ divisible by $x^4$ is $x^4z^2w^2$. 
By applying the same idea for $y,z$ and $w$ 
we can write 
\begin{equation}\label{form_identity}
F(x,y,z,w)=H(x,y,z,w) + G(x,y,z,w),
\end{equation}
where $H (x,y,z,w)=
a\,x^4z^2w^2+b\,y^4x^2w^2+c\,z^4x^2y^2+d\,w^4y^2z^2+e\,x^2y^2z^2w^2$ and 
$G(x,y,z,w)=F(x,y,z,w)-H(x,y,z,w)$. 
From the identity~\eqref{form_identity}, we see that every monomial in $G$ contains at least one variable 
on which the degree of the monomial is odd. 
We write $G(x,y,z,w)=\gamma_{x,3}(y,z,w)x^3+
\gamma_{x,2}(y,z,w)x^2+\gamma_{x,1}(y,z,w)x+\gamma_{x,0}(y,z,w)$ and examine  $\gamma_{x,3}$.

\begin{lemma} \label{lem1}
$\gamma_{x,3}(y,z,w)$ doesn't have monomials 
$y^3w^2$, $y^3zw$, $y^3z^2$, $y^2w^3$,  $y^2z^2w$, $y^2z^3$, $yzw^3$, $z^3 w^2$
 and $z^2w^3$. Thus
$\gamma_{x,3}(y,z,w)$ has only the monomials 
$y^2zw^2, \ yz^3w$ and $yz^2w^2$. 
\end{lemma}
\begin{proof}
From the inequality 
\[
F(y^2,y,z,w) \le O_{(4,1)}(y^2,y,z,w)=(1+z^2)w^2y^8+(z^2-4w^2)z^2y^6+y^2z^2w^4,
\]
we know that $\gamma_{x,3}(y,z,w)$ (briefly, $\gamma_{x,3}$) doesn't have the monomials 
$y^3w^2$, $y^3zw$, $y^3z^2$, $y^2w^3$, $y^2z^3$ by considering the highest degree. To see that  $\gamma_{x,3}$ doesn't have the monomial $y^2z^2w$, we divide both sides of the above inequality by $y^8$ and take limit as $y\to \infty$, and then divide both sides by $w^2$ and take limit as $w\to 0$.

Inequality $F(x,y,z,x^2) \le O_{(4,1)}(x,y,z,x^2)$ implies that $\gamma_{x,3}$ doesn't have monomials 
$z^2w^3$ and $yzw^3$. 
From the inequality $\lim_{y\to 0}y^2F(x,y,1/y,w) \le \lim_{y\to 0} y^2 O_{(4,1)}(x,y,1/y,w)$, we also see that $\gamma_{x,3}$ doesn't have the mononial $z^3w^2$. 

Consequently, we have  $\gamma_{x,3}(y,z,w)=yzw(q_{11}z^2+q_{12}zw+q_{13}yw)$
for some $q_{11},\ q_{12},\ q_{13} \in \mathbb{R}$.  
\end{proof}

Like the previous lemma~\ref{lem1}, we can check which monomials do not appear 
in $\gamma_{y,3}(x,z,w)$, $\gamma_{z,3}(x,y,w)$, and 
$\gamma_{w,3}(x,y,z)$. That is, we can easily see that 
\[
\begin{aligned}
\gamma_{y,3}(x,z,w)=&\,xzw(q_{21}xz+q_{22}xw+q_{23}w^2),\\
\gamma_{z,3}(x,y,w)=&\,xyw(q_{31}x^2+q_{32}xy+q_{33}yw),\\
\gamma_{w,3}(x,y,z)=&\,xyz(q_{41}xz+q_{42}yz+q_{43}y^2),
\end{aligned}
\]
for some $q_{ij}\in \mathbb R$.
From the above identities on $\gamma_{\cdot,3}$ and \eqref{form_identity}, we have that
\begin{multline}\label{form_G}
G(x,y,z,w)=xyzw \left( xz(s_1 y^2+ s_2 w^2)+yw(s_3 x^2+s_4 z^2) \right.
\\ \left. +(s_5 x^2z^2+s_6 y^2w^2)
+ s_7 xyz^2 + s_8 yzw^2 + s_9 zwx^2 + s_{10} wxy^2 \right),
\end{multline}
where $s_i \in \mathbb{R}$. 

\begin{lemma} \label{lem2}
We have $H=aO_{(4,1)}$ in the identity~\eqref{form_identity}. 
\end{lemma}
\begin{proof}
We note that $\sum G(x, \epsilon_1 y ,\epsilon_2 z, \epsilon_3 w)=0$, where the sum is  taken over all values $\epsilon_i\in \{-1,1\}$. So we have 
\[
H(x,y,z,w)=\frac18 \sum_{\epsilon_i \in \{ 1,-1 \}} 
F( x,\epsilon_1 y,\epsilon_2 z, \epsilon_3 w).
\] 
Therefore, we see that $0\le H \le O_{(4,1)}$.
Now, $O_{(4,1)}(x,x,x,x)=0$ implies that  
\[
H(x,x,x,x)=x^8(a+b+c+d+e)=0\Longrightarrow -e=a+b+c+d.
\]
Then, we get $H(x,x,z,z)=x^2z^2(x-z)(x+z)(bx^2-dz^2) \ge 0$, and so $b=d$. In a similar way, we can show that $a=b=c=d$. Consequently, we have $H=aO_{(4,1)}$ with $0\le a\le 1$. 
\end{proof}

From the Lemma~\ref{lem2}, the identity~\eqref{form_identity} is reduced to  
\begin{multline}\label{form_F}
F(x,y,z,w)=aO_{(4,1)}(x,y,z,w)+xyzw \left[ xz(s_1 y^2+ s_2 w^2)+yw(s_3 x^2+s_4 z^2) \right.
\\  \left.+(s_5 x^2z^2+s_6 y^2w^2)
+ s_7 xyz^2 + s_8 yzw^2 + s_9 zwx^2 + s_{10} wxy^2 \right] .
\end{multline}
To arrive at the goal $O_{(4,1)}\in \mathcal E(\mathcal P_{4,8})$, we need two more Lemmas.

\begin{lemma} \label{Lem3}
For $0\le a \le 1$, if $F_1=aO_{(4,1)}+x^2yz^2 w (\alpha y^2 +\beta w^2)\in \mathcal P_{4,8}$, 
then $\alpha =\beta =0$. 
\end{lemma}
\begin{proof}
Because $F_1(x,x,x,x)=(\alpha+\beta)x^8$ and 
$F_1(x,x,x,-x)=-(\alpha+\beta)x^8$, we see that $\alpha+\beta=0$. 
Then, $F_1(z,z,z,w)=z^4(z^2-w^2)[a(z^2-w^2)+\alpha zw] \ge 0$ implies $\alpha=0$.
\end{proof}
By the same arguement, we have the follwoing.
\begin{lemma} 
For $0\le a \le 1$, if $F_1=aO_{(4,1)}+xyzw (\alpha x^2 z^2 +\beta y^2 w^2)\in \mathcal P_{4,8}$,
then $\alpha =\beta =0$.  
\end{lemma}
Now, we can show that $O_{(4,1)}\in \mathcal E(\mathcal P_{4,8})$.
\begin{theorem} 
The quaternary octic $O_{(4,1)}(x,y,z,w)$ is extremal, i.e., $O_{(4,1)}\in \mathcal E (\mathcal P_{4,8})$.
\end{theorem}
\begin{proof}
Suppose $0\le F\le O_{(4,1)}$ and define a form
\[
F_{yw}(x,y,z,w)=\frac14 
\sum_{\epsilon_i \in \{ 1,-1 \}} 
F(\epsilon_1 x, y,\epsilon_2 z, w).
\]
Then, from the identity~\eqref{form_F}, we see that 
\[
F_{yw}=aO_{(4,1)}+x^2yz^2w(s_1 y^2+ s_2 w^2) \quad {\rm and }\quad 0\le F_{yw} \le O_{(4,1)}. 
\]
Then by Lemma~\ref{Lem3}, 
$s_1=s_2=0$. In a similar way, we can show 
$s_i=0\ (1\le i \le 10)$. So we have $F=aO_{(4,1)}$. This completes the proof.
\end{proof}

Now, we  define  quaternary octic form  $O^{\prime}_{(4,1)}$  and senary quartic form $Q_{(4,1)}$ by 
\[
\begin{aligned}
O_{(4,1)}^{\prime}(x,y,z,w)&=w^8+x^4 y^2 z^2 +y^4 x^2 z^2 +z^4 x^2 y^2 -4 x^2 y^2 z^2 w^2,\\
Q_{(4,1)}(p,q,s,t,u,v)&=B_{\Phi^{(4,1)}}\left(\begin{array}{cccc} p & s & u & v\\q & t & v & u\end{array}\right)\\
                    &=v^4+2(p^2 q^2+s^2 t^2 +u^2 v^2)+q^2 s^2+t^2 u^2+p^2 u^2\\
&\quad \quad \quad -2pqst-4pquv-4stuv.
\end{aligned}
\]
Since $\Phi^{(4,1)}$ is a positive linear map, we see that $Q_{(4,1)}$ is a psd form, that is, $Q_{(4,1)}\in \mathcal P_{6,4}$. From the arithmetic-geometric inequality, we can show that  $O^{\prime}_{(4,1)}\in \mathcal P_{4,8}$.
We also note that 
\begin{eqnarray}
O'_{(4,1)}(xz^2,xyw,zw^2,xzw)=x^4 z^6 w^6 O(x,y,z,w),\label{psd:o'}\\
Q_{(4,1)}(zw^3,xyw^2,yzw^2,xw^3,xyzw,w^4)=w^8 O^{\prime}(x,y,z,w).\label{psd:q}
\end{eqnarray}

Now, we show that $O^{\prime}_{(4,1)}\in \mathcal P_{4,8}$ is extremal psd. Suppose $O^{\prime}_{(4,1)}\ge F\in \mathcal P_{4,8}$. Then we have
\[
x^4 z^6 w^6 O_{(4,1)}(x,y,z,w)=O'_{(4,1)}(xz^2,xyw,zw^2,xzw)\ge F(xz^2,xyw,zw^2,xzw)\ge 0.
\]
Since the left-hand side is extremal, we must have 
\begin{equation}\label{ext_op}
F(xz^2,xyw,zw^2,xzw)=\alpha O'_{(4,1)}(xz^2,xyw,zw^2,xzw)
\end{equation}
for some $\alpha \in \mathbb R$. We replace $x,y,z$ and $w$ by 
\[
\left( \frac{w^4}{z^2 x}\right)^{\frac 13},\ 
\left( \frac{x^2 y^3 z}{w^5}\right)^{\frac 13},\ 
\left( \frac{x^2 z}{w^2}\right)^{\frac 13},\ {\rm and }\
\left( \frac{z w}{x}\right)^{\frac 13},
\]
respectively. Then Eq.~\eqref{ext_op} becomes $F(x,y,z,w)=\alpha O^{\prime}_{(4,1)}(x,y,z,w)$.
This completes the following Proposition.
\begin{proposition}
The quaternary octic $O^{\prime}_{(4,1)}(x,y,z,w)$ is extremal, i.e., $O^{\prime}_{(4,1)}\in \mathcal E (\mathcal P_{4,8})$.
\end{proposition}
We proceed to examime the extremality of $Q_{(4,1)}\in \mathcal P_{6,4}$. 
From the Eq.~\eqref{psd:q} and the extremality of $O^{\prime}_{(4,1)}$, it follows that whenever $Q_{(4,1)}\ge F\in \mathcal P_{6,4}$, we have 
\begin{equation}\label{ext_q}
F(zw^3,xyw^2,yzw^2,xw^3,xyzw,w^4)=\alpha Q_{(4,1)}(zw^3,xyw^2,yzw^2,xw^3,xyzw,w^4)
\end{equation}
for some $\alpha\in \mathbb R$. Replacing $x,y,z$ and $w$ by 
\[
x=t v^{-3/4},\quad y=q v^{1/4}t^{-1},\quad z=p v^{-3/4},\quad w=v^{1/4}
\]
respectively, Eq.~\eqref{ext_q} becomes 
\begin{equation}\label{ext_qp}
F(p,q,pq/t,t,pq/v,v)=\alpha Q_{(4,1)}(p,q,pq/t,t,pq/v,v).
\end{equation}
We consider $G:=F-\alpha Q_{(4,1)}$. Since $G(p,q,pq/t,t,pq/v,v)=0$, we see that $G$ is of the form
\[
G(p,q,s,t,u,v)=(pq-st)G_1(p,q,s,t,u,v)+(pq-uv)G_2(p,q,s,t,u,v).
\]
Using the equality $Q_{(4,1)}(p,q,s,t,u,v)=Q_{(4,1)}(t,s,q,p,u,v)$ and looking at the leading coefficient of each variable, we can get
\[
\begin{aligned}
G=&a_1(stv^2-2uv^3+pqv^2)+a_2(-uv^2 q-uv^2 s+pqsv+stqv)\\
   &+a_3(-2u^2v^2+stuv+pquv)+a_4(pqst-u^2v^2)+a_5(-2qsuv+pq^2s+s^2tq)\\
  &+a_6(p^2qu-pu^2v-tu^2v+st^2u)+a_7(p^2q^2+s^2t^2-stuv-pquv).
\end{aligned}
\]
Note that $Q_{(4,1)}(p,q,s,t,u,v)=0$ on the set $S=\{(p,q,s,t,u,v)|pq=st,\,pq=uv,\,st=uv\}$. Thus $F=0$ on the set $S$ and so $F$ has local minima on the set S. From $\partial F/\partial p=\partial (G+\alpha Q_{(4,1)})/\partial p=0$ on the set $S$, we get $a_1=a_2=a_5=a_6=0$ and $a_3=-a_4-a_7$. Compute $F=\alpha G_{(4,1)}$ when $u=1/v$, $t=v^3$, $s=1/v^3$ and $p=q=0$, then we get $a_7=0$. 

Therefore, we have that 
\[ 
0\le F=a_4(st-uv)(pq-uv)+\alpha Q_{(4,1)}\le Q_{(4,1)}\quad {\rm with }\ 0\le \alpha \le 1.
\]
Now, we will show that $a_7=0$ from the condition $F\ge 0$. When $t=v^2$ and $u=1$, the discriminant $D(F,p)$ of $F$ on the variable $p$ should be less than or equal to $0$. That is, $-D(F,p)\ge 0$. From the identity
\[
\begin{aligned}
-D(F,p)=(&4\alpha^2 q^2+8\alpha^2 q^4+8\alpha^2 v^4 -a_7^2 q^2 v^4+4a_7 \alpha q^2 v^4+12\alpha^2 q^2v^4)s^2\\
-&(4a_7\alpha v^3+16\alpha^2 v^3-2a_7^2 q^2 v^3+4a_7\alpha q^2 v^3+48\alpha^2 q^2 v^3)s\\
&+(4a_7\alpha v^2+8\alpha^2v^2-a_7^2 q^2 v^2+8\alpha^2v^4+16\alpha^2 q^2 v^4),
\end{aligned}
\]
we compute the condition on which the discriminant $D(-D(F,p),s)$ should be less than or equal to $0$.
Since the coefficient of the highest degree of $q$, $q^6$, is $32\alpha^2 v^2 (a_7^2-16\alpha^2v^2)$ in $D(-D(F,p),s)$. it follows that  $32\alpha^2 v^2 (a_7^2-16\alpha^2v^2)\le 0$ for all $v\neq 0$. Consequently, we have $a_7=0$. This completes the proof of the following Proposition.
\begin{proposition}\label{psd_form}
The senary quartic $Q_{(4,1)}(p,q,s,t,u,v)$ is extremal, i.e., $Q_{(4,1)}\in \mathcal E (\mathcal P_{6,4})$.
\end{proposition}

\section{Extrmality for Qi and Hou's map}
In this section, we show that $\Phi^{(n,k)}$ is not extremal whenever $n$ and $k$ are not relatively prime, that is, $\gcd(n,k)>1$. 
This answers the question on the existence of optimal EW without extremality as well as spanning property.
For the case of $\gcd(n,k)=1$, we think that $\Phi^{(n,k)}$ may be extremal. Although our proof can be applicable for general case, it is too laborious. So we give the details of the proof for extremality of $\Phi^{(4,1)}$ and $\Phi^{(4,3)}$.

We begin with showing that $B_{\Phi^{(n,k)}}$ is not extremal whenever $\gcd(n,k)>1$.
Let $S_n$ be the symmetric group consisting of all bijection (permutation) from the set $\{1,2,\cdots,n\}$ onto itself.
For any integer $q$, define $\sigma_q \in S_n$ by $$\sigma_q(j)\equiv j+q \pmod{n} .$$ 
First, we consider the case when $k$ divides $n$ and $({n}/{k}) >1$. 
Note that $\sigma_k$ is a product of disjoint $k$ cycles. 
We recall  the  biquadratic form $B_{\Phi^{(n,k)}}$ 
\begin{equation}\label{dec_biquad1}
B_{\Phi^{(n,k)}}=(n-2)\sum_{i=1}^n x_i^2 y_i^2 -2 \sum_{1\le i<j\le n} x_iy_ix_jy_j
+\sum_{i=1}^n x_{\sigma_k(i)}^2y_i^2,
\end{equation}
and define biquadratic forms 
\begin{equation}\label{decomp}
F_{\sigma_k,d}=\left ( \frac{n}{k}-2 \right)\sum_{i \equiv d \hskip -5pt \pmod{k}} x_i^2 y_i^2 
-2 \sum_{\substack{1\le i<j\le n \\ i \equiv j\equiv d \hskip -5pt \pmod{k} }} x_iy_ix_jy_j
+\sum_{i \equiv d \hskip -5pt \pmod{k}} x_{\sigma_k(i)}^2y_i^2 
\end{equation}
for each $d=1,\cdots,k$. Then we can easily check that 
\begin{equation}\label{dec_biquad2}
B_{\Phi^{(n,k)}}=\sum_{d=1}^{k} F_{\sigma_k,d}+
\sum_{\substack{1\le i<j\le n \\ i \not\equiv j \hskip -5pt \pmod{k} }} 
\left( x_i y_i -x_j y_j \right)^2 .
\end{equation}
Now, we see that all the biquadratic forms  $F_{\sigma_m ,d}$'s in~\eqref{dec_biquad2} are equivalent to the biquadratic form $B_{\Phi^{(n/k,1)}}$. That is, by renaming $x_{d+i\,m}$ by $x_{i+1}$ and $y_{d+i\,m}$ by $y_{i+1}$ in the biquadratic form $F_{\sigma_k , d}$, we get the $B_{\Phi^{(n/k,1)}}$
\begin{equation}\label{dec_biquad3}
B_{\Phi^{(n/k,1)}} = \left(\frac n k-2\right)\sum_{i=1}^{n/k} x_i^2 y_i^2 -2 \sum_{1\le i<j\le n/k} x_iy_ix_jy_j
+\sum_{i=1}^{n/k} x_i^2y_{\sigma_1(i)}^2,
\end{equation}
where $\sigma_1$ is a permutation in $S_{n/k}$ defined by $\sigma_1(j)=j+1$.
Thus $F_{\phi_m,d}$'s are equivalent to $B_{\Phi^{(n/k,1)}}$. Furthermore, we can conclude that each $F_{\sigma_k , d}$ in~\eqref{dec_biquad2} is positive semidefinite quadratic form since $B_{\Phi^{(n/k,1)}}$ is psd. Consequently, we have the following result.
\begin{proposition} \label{prop:divide}
If $k$ divide $n$, then psd biquadractic form  $B_{\Phi^{(n,k)}}$ is decomposed as a sum of psd biquadratic forms as in the identity~\eqref{dec_biquad2}. Futhermore, each $F_{\sigma_k, d}$ in~\eqref{dec_biquad2} can be considered as $B_{\Phi^{(n/k,1)}}$ by renaming.
\end{proposition}

Now, we assume $\gcd(n,q)=k \ge1$. Then we write $q=k m$ such that $\gcd(n/k,m)=1$. 
Define $\mu \in S_n$ by $\mu(d+k\,j) \equiv d+q j \equiv d+k\,m\,j \pmod{n}$. 
Then it is easy to check $\mu$ is well-defined and 
$\mu \circ \sigma_k = \sigma_q \circ \mu$. To represent renaming, we define $\mu(B_{\Phi^{(n,k)}})$ by 
$$
\mu(B_{\Phi^{(n,k)}})=(n-2)\sum_i x_{\mu(i)}^2 y_{\mu(i)}^2 -
2 \sum_{1\le i<j\le n} x_{\mu(i)}y_{\mu(i)}x_{\mu(j)}
y_{\mu(j)} +\sum_i x_{\mu(\sigma_k(i))}^2y_{\mu(i)}^2 .  
$$
Then we see that
\[
\mu(B_{\Phi^{(n,k)}})=(n-2)\sum_i x_i^2 y_i^2 -2 \sum_{1\le i<j\le n} x_iy_ix_jy_j
+\sum_i x_{\sigma_q(i)}^2y_i^2 
=B_{\Phi^{(n,q)}}. 
\]
Therefore,  we have the following Proposition.
\begin{proposition}\label{prop:gcd}
Let $\gcd(n,q)=k\ge 1$. Then, $B_{\Phi^{(n,q)}}$ is extremal if and only if $B_{\Phi^{(n,k)}}$ is extremal.
\end{proposition}
Finally, we can show that the main result.
\begin{theorem}\label{thm:main}
If $\gcd(n,k)\neq 1$, then $\Phi^{(n,k)}$ is not extremal positive linear map.
\end{theorem}
\begin{proof}
From the Proposition~\ref{prop:gcd}, it suffices to consider the case when $k$ divide $n$. In this case, we know that biquadratic form $B_{\Phi^{(n,k)}}$ is the sum of positive semidefinite biquadractic forms from the proposition~\ref{prop:divide}. Therefore the corresponding map $\Phi^{(n,k)}$ is the sum of positive linear maps as explained in the last paragraph of the section~\ref{map_form}. That is, $\Phi^{(n,k)}$ is not extremal.
\end{proof}
As a byproduct, we have the following corollaries.
\begin{corollary}There exists an optimal EW which neither has no spanning property nor is associated extremal positive linear map. In fact, $W_{\Phi^{(n,k)}}$ is such an optimal EW whenever $\gcd(n,k)\neq 1$.
\end{corollary}
\begin{corollary}\label{coro:equiv}
If $\gcd(n,k)=1$, then $\Phi^{(n,k)}$ is extremal if and only if $\Phi^{(n,1)}$ is extremal
\end{corollary}
\begin{proof}
From the proof of Theorem~\ref{thm:main}, we see not non-extremality of $B_{\Phi^{(n,k)}}$ implies the non-extremality of $\Phi^{(n,k)}$. By combining the results of Proposition~\ref{prop:divide},\ref{prop:gcd} and Theorem~\ref{thm:main}, the proof is completed.
\end{proof}
\begin{corollary} A positive linear map $\Phi^{(n,n/2)}$ is decomposable when $n$ is even natural number greater than $2$.
\end{corollary}
\begin{proof} Since $n/2$ divide $n$, we have the decomposition of $B_{\Phi^{(n,n/2)}}$ as in~\eqref{dec_biquad2}. We also see that each $F_{\sigma_{n/2},d}$ in~\eqref{decomp} is of the form
\[
\begin{aligned}
F_{\sigma_{n/2},d}=&-2\,x_d\,y_d\,x_{(d+n/2)}\, y_{(d+n/2)}+x_d^2\, y_{(d+n/2)}^2+x_{(d+n/2)}^2\, y_d^2\\
=&(x_d\,y_{d+n/2}-x_{d+n/2}\,y_d)^2.
\end{aligned}
\]
Since the positive linear map corresponding to the positive semidefinite biquadratic form $(x_d\,y_{d+n/2}-x_{d+n/2}\,y_d)^2$ is completely copositive and the map corresponding to $(x_i\,y_i-x_j\,y_j)^2$ is completely positive, we can conclude that positive linear map $\Phi^{(n,n/2)}$ is decomposable.
\end{proof}

We now turn to the extremality of $\phi^{(4,k)}$. In this case, we can show that $\Phi^{(4,1)}$ and $\Phi^{(4,3)}$ are extremal  from the extremality of the senary quartic form $Q_{(4,1)}$ (Recall the proposition~\ref{psd_form}).
\begin{theorem}
$\Phi^{(4,k)}$ is an extremal positive linear map if and only if $k=1$ or $3$.
\end{theorem}
\begin{proof}
From the theorem~\ref{thm:main}, we know that $\Phi^{(4,2)}$ is not extremal. We also know that $\Phi^{(4,1)}$ is extremal if and only if $\Phi^{(4,3)}$ is extremal by Corollary~\ref{coro:equiv}. Therefore, 
It suffices to show that $B_{\Phi^{(4,1)}}$ is an extremal positive semidefinite biquadratic form as stated in section~\ref{map_form}. 
Suppose $F$ is a biquadratic form such that $B_{\Phi}\ge F\ge 0$. Then 
\[
Q_{(4,1)}(p,q,s,t,u,v)=B_{\Phi^{(4,1)}}\left(\begin{array}{cccc} p & s & u & v\\q & t & v & u\end{array}\right)\ge 
F\left(\begin{array}{cccc} p & s & u & v\\q & t & v & u\end{array}\right)\ge 0.
\]
From the extremeness of $Q_{(4,1)}$, we have 
\begin{equation}\label{biquad1}
F\left(\begin{array}{cccc}  p & s & u & v\\q & t & v & u\end{array}\right) =\lambda_1 Q_{(4,1)}(p,q,s,t,u,v).
\end{equation}
Since $B_{\Phi^{(4,1)}}$ is invariant under the cyclic permutation $(1234)$ applied to the subscipts of $(x_1,x_2,x_3,x_4)$ and $(y_1,y_2,y_3,y_4)$ simultaneously, we see that 
\begin{equation}\label{biquad2}
\begin{aligned}
&Q_{(4,1)}(p,q,s,t,u,v)\\
=&B_{\Phi^{(4,1)}}\left(\begin{array}{cccc}p & s & u & v\\q & t & v & u\end{array}\right)
=B_{\Phi^{(4,1)}}\left(\begin{array}{cccc} s & u & v & p\\t & v & u & q\end{array}\right)
=B_{\Phi^{(4,1)}}\left(\begin{array}{cccc} u & v & p & s\\v & u & q & t\end{array}\right)
=B_{\Phi^{(4,1)}}\left(\begin{array}{cccc} v & p & s & u\\u & q & t & v\end{array}\right).
\end{aligned}
\end{equation}
So we can similarly show that 
\begin{equation}\label{biquad3}
\begin{aligned}
F\left(\begin{array}{cccc}  s & u & v & p\\t & v & u & q\end{array}\right) &=\lambda_2 Q_{(4,1)}(p,q,s,t,u,v),\\
F\left(\begin{array}{cccc}  u & v & p & s\\v & u & q & t\end{array}\right)&=\lambda_3 Q_{(4,1)}(p,q,s,t,u,v),\\
F\left(\begin{array}{cccc}  v & p & s & u\\u & q & t & v\end{array}\right) &=\lambda_4 Q_{(4,1)}(p,q,s,t,u,v).\\
\end{aligned}
\end{equation}
By comparing the coefficients, we get $\lambda_1=\lambda_2=\lambda_3=\lambda_4$. 
In fact, we see that
\[
Q_{(4,1)}(s,t,t,s,s,t)=Q_{(4,1)}(t,s,s,t,t,s)=2(s^2-t^2)^2,
\]
and we get the following indentities
\[
\begin{aligned}
&F\left(\begin{array}{cccc}
s & t & s & t\\
t & s & t & s\end{array}\right)\\
=& \lambda_1Q_{(4,1)}(s,t,t,s,s,t)=\lambda_2Q_{(4,1)}(t,s,s,t,t,s)=\lambda_3Q_{(4,1)}(s,t,t,s,s,t)=\lambda_4Q_{(4,1)}(t,s,s,t,t,s)
\end{aligned}
\]
from \eqref{biquad1} and \eqref{biquad3}. This give rise to $\lambda_1=\lambda_2=\lambda_3=\lambda_4$.

Now, for any fixed nonzero real numbers $y_1,y_2,y_3$ and $y_4$, we define a quadratic form $f(x_1,x_2,x_3,x_4)$ by 
\[
f(x_1,x_2,x_3,x_4):=(F-\lambda_1 B_{\Phi^{(4,1)}})\left(\begin{array}{cccc}x_1 & x_2 & x_3 & x_4\\y_1 & y_2 & y_3 & y_4\end{array}\right).
\]
From the identities \eqref{biquad2} and \eqref{biquad3}, we see that 
\[
f(x_1,x_2,y_4,y_3)\equiv 0,\ f(x_1,y_3,y_2,x_4)\equiv 0,\ f(y_2,y_1,x_3,x_4)\equiv 0,\ f(y_4,x_2,x_3,y_1)\equiv 0.
\]
Note that $f(x_2,x_2,y_4,y_3)\equiv 0$ implies that $f$ is divisible by $y_3 x_3 -y_4 x_4$. Similarly, we see that $f$ is divisible by $y_2 x_2-y_3 x_3$, $y_1 x_1 -y_2 x_2$ and $y_1 x_1-y_4 x_4$. Since degree of $f$ is $2$, this leads to $f\equiv 0$. In other words, we have 
\[
(F-\lambda_1 B_{\Phi^{(4,1)}})\left(\begin{array}{cccc}x_1 & x_2 & x_3 & x_4\\y_1 & y_2 & y_3 & y_4\end{array}\right)\equiv 0
\]
whenever $y_1,y_2,y_3$ and $y_4$ are nonzero real numbers. By continuity, we conclude that 
$F=\lambda B_{\Phi}$. Therefore, $B_{\Phi}$ is extremal among biquadratic forms.
\end{proof}

\section{Conclusion}
In this paper, we have studied the extremality of the positive linear map $\Phi^{(n,k)}$  constructed by Qi and Hou \cite{xia11}, those associated entanglement witnesses $W_{\Phi^{(n,k)}}$'s are known \cite{xia} as optimal indecomposable entanglement witnesses without spanning property.
One of the interesting problems on optimal indecomposable entanglement witnesses is whether an optimal indecomposable witness $W$ exists such that the associated positive linear map is not extremal and corresponding $\mathcal P_W$ do not span the Hilbert space fully. Here, we answer this question negatively by showing that  $\Phi^{(n,k)}$ is not extremal whenever $\gcd(n,k)\neq 1$.  As a byproduct of our proof using the correspondece between positive semidefinite biquadratic forms and positive linear maps, we have reproved that $\Phi^{(n,n/2)}$ is decomposable when $n$ is even.

For the case of $\gcd(n,k)=1$, we showed that $\Phi^{(n,k)}$ is extremal if and only if $\Phi^{(n,1)}$ is extremal. In particular, we proved that $\Phi^{(4,1)}$ and $\Phi^{(4,3)}$ are extremal when $n=4$. 
Our proof for the extremality seems to be applicable for general $(n,k)$ with $\gcd(n,k)=1$. But, it is too laborious since we should check the extremality of each $B_{\Phi^{(n,1)}}$. So  a new approach which can be applicable for all cases at the same time is needed.

By the way, Chru\'{s}ci\'{n}ski and  Wudarski \cite{cw12} gave another variant $\Phi[a,b,c,d]$ of extremal Choi map between $M_4$ recently. Entanglement witnesses arising from these maps are known to be indecomposable optimal entanglement witnesses which have both spanning property and co-spanning property. We say that an entanglement witness $W$ has co-spanning property if $W^{\Gamma}$ has spanning property \cite{kye_ritsu,hakye12}. Among them, $\Phi[1,1,1,0]$ and $\Phi[1,0,1,1]$ are expected to be extremal. But, in this case, our method is not  directly applicable since the corresponding quaternary octic psd form $O_{\Phi}$ is not extremal. Therefore, it would be interesting to investigate the extremality of these maps. We also note that the extremality of maps $\Phi[1,0,p_{\theta}-1;\theta],\,\Phi[1,p_{\theta}-1,0;\theta]$ \cite{hakye12pra} is open question so far.

\begin{acknowledgments}
This work was partially supported by the Basic Science Research Program through the
National Research Foundation of Korea(NRF) funded by the Ministry of Education, Science
and Technology (Grant No. NRFK 2012-0002600 to K. -C. Ha and Grant No. NRFK 2011-0026832 to H. -S. Yu)
\end{acknowledgments}


\end{document}